\documentclass{llncs}
\usepackage{amsmath}
\usepackage{amssymb}
\usepackage{gastex}

\DeclareSymbolFont{rsfscript}{OMS}{rsfs}{m}{n}
\DeclareSymbolFontAlphabet{\mathrsfs}{rsfscript}

\newcommand{\sw}{reset word}
\newcommand{\sws}{reset words}
\newcommand{\sa}{synchronizing automata}
\newcommand{\san}{synchronizing automaton}

\newtheorem{conj}{Conjecture}

\begin{document}

\title{Synchronizing Automata on \\Quasi Eulerian Digraph}

\titlerunning{Synchronizing Automata on \\Quasi Eulerian Digraph}

\author{Mikhail V. Berlinkov}

\authorrunning{M. V. Berlinkov}

\tocauthor{M.V.Berlinkov (Ekaterinburg, Russia)}

\institute{Laboratory of Combinatorics, \\Institute of Mathematics
and Computer Science,\\Ural Federal University,\\
620083 Ekaterinburg, Russia\\
\email{m.berlinkov@gmail.com}}

\maketitle

\begin{abstract}
In 1964 \v{C}ern\'{y} conjectured that each $n$-state synchronizing
automaton posesses a reset word of length at most $(n-1)^2$. From
the other side the best known upper bound on the reset length
(minimum length of reset words) is cubic in $n$. Thus the main
problem here is to prove quadratic (in $n$) upper bounds. Since
1964, this problem has been solved for few special classes of \sa.
One of this result is due to Kari~\cite{Ka03} for automata with Eulerian digraphs.
In this paper we introduce a new approach to prove quadratic
upper bounds and explain it in terms of Markov chains and Perron-Frobenius theories.
Using this approach we obtain a quadratic upper bound for a generalization of Eulerian automata.
\end{abstract}

\section{Synchronizing automata and the \v{C}ern\'y conjecture}

Suppose $\mathrsfs{A}$ is a complete deterministic finite automaton
whose input alphabet is $\Sigma$ and whose state set is $Q$. The
automaton $\mathrsfs{A}$ is called \emph{synchronizing} if there
exists a word $w\in\Sigma^*$ whose action \emph{resets}
$\mathrsfs{A}$, that is, $w$ leaves the automaton in one particular
state no matter at which state in $Q$ it is applied: $q.w=q'.w$ for
all $q,q'\in Q$. Any such word $w$ is called \emph{reset} (or
\emph{synchronizing}) for the automaton. The minimum length of reset
words is called \emph{reset length} and can be denoted by
$\mathfrak{C}(\mathrsfs{A})$.

Synchronizing automata serve as transparent and natural models of
error-resistant systems in many applications (coding theory,
robotics, testing of reactive systems) and also reveal interesting
connections with symbolic dynamics and other parts of mathematics.
For a brief introduction to the theory of \sa\ we refer the reader
to the recent survey~\cite{Vo08}. Here we discuss one of the main
problems in this theory: proving an upper bound of magnitude
$O(n^2)$ for the minimum length of reset words for $n$-state \sa.

In~1964 \v{C}ern\'{y}~\cite{Ce64} constructed for each $n>1$ a \san\
$\mathrsfs{C}_n$ with $n$ states whose shortest \sw\ has length
$(n-1)^2$, i.e. $\mathfrak{C}(\mathrsfs{C}_n)=(n-1)^2$.
The automaton $\mathrsfs{C}_4$ is drawn on figure~\ref{fig:4rcp}.
Soon after that he conjectured that those automata represent the worst possible
case, thus formulating the following hypothesis:
\begin{conj}[\v{C}ern\'y]
\label{Cerny_Conj} Each \san\ $\mathrsfs{A}$ with $n$ states has a
\sw\ of length at most \makebox{$(n-1)^2$}, i.e.
$\mathfrak{C}(\mathrsfs{A}) \leq (n-1)^2$.
\end{conj}
By now this simply looking conjecture is arguably the most
longstanding open problem in the combinatorial theory of finite
automata. Moreover, the best upper bound known so far is due to
Pin~\cite{Pi83}\footnote{An upper bound of order
$\Omega(\frac{7n^3}{48})$ has been proved in \cite{TR_7_48}. But we
know about one unclear place in the proof of this result.} (it is based
upon a combinatorial theorem conjectured by Pin and then proved by
Frankl~\cite{Fr82}): for each \san\ with $n$ states, there exists a
\sw\ of length $\frac{n^3-n}6$. Since this bound is cubic and the
\v{C}ern\'y conjecture claims a quadratic value, it is of certain
importance to prove quadratic (upper) bounds for some classes of
\sa.

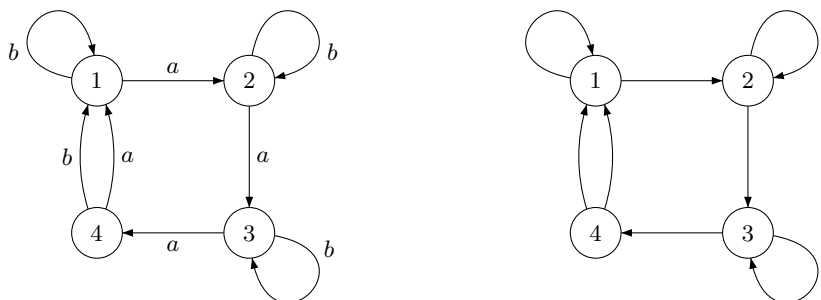
\begin{figure}[ht]
 \begin{center}
  \unitlength=3.2pt
 \begin{picture}(18,26)(20,-2)
    \gasset{Nw=6,Nh=6,Nmr=3}
    \thinlines
    \node(A1)(0,18){$1$}
    \node(A2)(18,18){$2$}
    \node(A3)(18,0){$3$}
    \node(A4)(0,0){$4$}
    \drawloop[loopangle=135,ELpos=30](A1){$b$}
    \drawloop[loopangle=45,ELpos=70](A2){$b$}
    \drawloop[loopangle=-45,ELpos=25](A3){$b$}
    \drawedge(A1,A2){$a$}
    \drawedge(A2,A3){$a$}
    \drawedge(A3,A4){$a$}
    \drawedge[curvedepth=2](A4,A1){$b$}
    \drawedge[ELside=r, curvedepth=-2](A4,A1){$a$}
 \end{picture}
 \begin{picture}(18,26)(-20,-2)
    \gasset{Nw=6,Nh=6,Nmr=3}
    \thinlines
    \node(A1)(0,18){$1$}
    \node(A2)(18,18){$2$}
    \node(A3)(18,0){$3$}
    \node(A4)(0,0){$4$}
    \drawloop[loopangle=135](A1){}
    \drawloop[loopangle=45](A2){}
    \drawloop[loopangle=-45](A3){}
    \drawedge(A1,A2){}
    \drawedge(A2,A3){}
    \drawedge(A3,A4){}
    \drawedge[curvedepth=2](A4,A1){}
    \drawedge[ELside=r, curvedepth=-2](A4,A1){}
 \end{picture}
 \end{center}
 \caption{Automaton $\mathrsfs{C}_4$ and its underlying graph}
 \label{fig:4rcp}
\end{figure}

\section{Exponents of primitive matrices vs reset thresholds}

In the rest of the paper, we assume that $\mathrsfs{A}$ is a
synchronizing $n$-state automaton with $k$-letter input alphabet
$\Sigma=\{a_1,a_2, \dots ,a_k \}$ and whose state set is $Q$. We
also assume $n>1,k>1$ and $\mathrsfs{A}$ is strongly connected because finding
\sws\ of length $O(n^2)$ can be easily reduced to this case (see
\cite{Pi78} for example). Now let us consider relations between
primitive matrices and synchronizing automata. In order to do this
we determine a natural linear structure associated with automata. We
mean the states of $\mathrsfs{A}$ as numbers $1,2,\dots,n$ and then
assign to each subset $T\subseteq Q$ its \emph{characteristic
vector} $[T]$ in the linear space $\mathbb{R}^n$ defined as follows:
the $i$-th entry of $[T]$ is 1 if $i\in T$, otherwise it is equal to
0. As usually, for any two vectors $g_1,g_2 \in \mathbb{R}^n$ we
denote the inner product of these vectors by $(g_1,g_2)$.

A matrix $M$ is \emph{primitive} if it is non-negative and its
$m$-th power is positive for some natural number $m$. The minimum
number $m$ with this property is called an \emph{exponent} of the
matrix $M$ and can be denoted by $exp(M)$. Let us also define a
\emph{weak exponent} of matrix $M$ as a minimum number $m$ such that
$M^m$ has a positive row and denote it by $wexp(M)$. Note that a
(weak) exponent depends only on the set of indices with positive
elements $Supp(M) = \{(i,j) \mid M_{i,j}>0 \}$ and doesn't depend on
their values. So when we consider exponents of some matrix $M$ we
can assume it is a 1-0 matrix and there is a \emph{corresponding}
graph with the adjacency matrix $M$. Moreover, $M^t_{i,j}$ is equal to the
number of directed paths of length exactly $t$ from state $i$ to $j$ in the corresponding graph.
The following proposition shows the basic properties of primitive matrices.

\begin{proposition}\label{prop_exp}Let $M$ be $n \times n$ primitive matrix. Then
\begin{enumerate}
    \item \label{i3} $wexp(M) \leq exp(M) \leq wexp(M)+n-1$;
    \item \label{i4} $exp(M) \leq (n-1)^2 + 1$ and equality holds only for Wielandt
    matrices;
    \item \label{i5} $wexp(M) \leq (n-1)^2$.
\end{enumerate}
\end{proposition}
\begin{proof}
    The left part of the inequality in item~\label{i3} follows immediately from the
definitions. The right part can be easily proved in terms of graph theory.
Indeed, let $i$ be the index of a positive row in $M^{wexp(M)}$. Then
there is a path (in the corresponding graph) of length $d(i,t) \leq
n-1$ from the state $i$ to each state $t$. Further, for each state
$s$ there is a path to some state $q$ of length $n-1-d(i,t)$ and a
path of length $wexp(M)$ from $q$ to $i$. Thus there is a path $s
\rightarrow q \rightarrow i \rightarrow t$ of length $$n-1-d(i,t) +
wexp(M) + d(i,t)=wexp(M) + n-1$$ and the item is proved.
\par
Item~\ref{i4} has been proven by Wielandt~\cite{WIL}. Item~\ref{i5}
follows from the fact that $exp(W)=(n-1)^2+1$ only for
the Wielandt matrix $W$ (item~\ref{i4}) but $wexp(W) = n^2-3n+3 \le
(n-1)^2$.
\end{proof}

The proof of the following proposition can be found in \cite{VAG} but we introduce it here to
be self-contained.

\begin{proposition}\label{prop_prim}Let $UG(\mathrsfs{A})$ denotes the underlying graph of the automaton
$\mathrsfs{A}$ and $M=M(UG(\mathrsfs{A}))$ denotes its adjacency matrix. Then
\begin{enumerate}
    \item \label{i0} $M = \sum_{i=1}^{k}{[a_i]}$;
    \item \label{i1} $M$ is a primitive matrix;
    \item \label{i2} $wexp(M) \leq \mathfrak{C}(\mathrsfs{A})$.
\end{enumerate}
\end{proposition}
\begin{proof}
Item~\ref{i0} follows immediately from definitions. Since $\mathrsfs{A}$ is a \san\
there exists a reset word $w$ of length $\mathfrak{C}(\mathrsfs{A})$ which takes all the states of the
automaton to some state $i$. This means that $i$-th row of $M^{|w|}$
is positive, so item~\ref{i2} is proved. Since $\mathrsfs{A}$ is strongly connected item~\ref{i1} is true also.
\end{proof}

It follows from above propositions that weak exponent of the underlying graph of $\mathrsfs{A}$
is at most $(n-1)^2$. This means that there are (unlabeled) paths of equal length $l \leq (n-1)^2$ from every state
of the automaton $\mathrsfs{A}$ into some particular state. The \v{C}ern\'y conjecture asserts
additionally that such paths can be chosen to be labeled by some fixed (reset) word. It seems that this additional
demand should increase significantly the minimum length of such paths. Indeed, for a lot of \sa\ the
reset length is much more than the weak exponent of its underlying graph. For instance, if \sa\ contains a loop
then its weak exponent is at most $n-1$ but its reset length can be equal $(n-1)^2$ (for \v{C}ern\'{y} series).
However, in order to prove the \v{C}ern\'{y} conjecture we only need such bound in the worst case and in \cite{VAG}
a strong connection between distribution of reset lengths of \sa\ and exponents of primitive graphs is considered.

\section{Markov chains and an extension method}

The aim of this paper is to obtain upper bounds on reset lengths by utilizing its connection with exponents of primitive graphs.
Let we have probability vector $p \in R_{+}^{k}$ on $\Sigma$ naturally extended on words as $p(v) = \prod_{i=1}^{|v|}{p(v(i))}$.
Now consider a random process of walking some agent in the underlying graph
$G = UG(\mathrsfs{A})$ walking by arrow labeled by $a_i$ with probability $p(a_i)$.
Then the matrix $S(\mathrsfs{A},p)=\sum_{i=1}^{k}{p(a_i)*a_i}$ is a probability matrix of this Markov process.
Let us note that $Supp(S(\mathrsfs{A},p))=Supp(M(UG(\mathrsfs{A})))$ and $S(\mathrsfs{A},p)$ is also column stochastic.
To simplify our notations denote by $1_n$ a vector in $R^n$ with all components equals $\frac{1}{n}$.
The following proposition summarize properties of Markov chains that we need.
\begin{proposition}\label{prop_markov}Let $S$ be a column stochastic $n \times n$ primitive matrix of some Markov process. Then
\begin{enumerate}
    \item \label{m0} $1_n$ is a left eigenvector of $S$, i.e. $S^t 1_n = 1_n$;
    \item \label{m1} there exists a steady state distribution $\alpha = \alpha(S) \in R^{n}_{+}$ of this Markov process, i.e. $S \alpha = \alpha$ and $(\alpha,n 1_n)=1$;
    \item \label{m2} 1 is a unique modulo-maximal eigenvalue of $S$ and the corresponding eigenspace is one dimension;
\end{enumerate}
\end{proposition}
\begin{proof}
    Since $S$ is a column stochastic matrix then $S^t 1_n = 1_n$. Thus $1$ is an eigenvalue of $S$ and corresponding eigenvector $[Q]$ is positive.
Since $S$ also is primitive then by Perron-Frobenius theorem $1$ is a unique modulo-maximal eigenvalue of $S$ and there is also
a unique (right) positive eigenvector $\alpha$, i.e. $S \alpha = \alpha$. Note that $\alpha$ can be chosen to be stochastic.
Also by Perron-Frobenius right and left eigenspaces corresponding to the eigenvalue $1$ are one dimension and equals to $<1_n>, <\alpha>$
respectively.
\end{proof}

For $K \subseteq Q$ and $v \in \Sigma^*$ we denote by $K.v$ and
$K.v^{-1}$ the image and respectively the preimage of the subset $K$
under the action of the word $v$, i.e.
$$K.v=\{q.v \mid q \in K\} \text{ and } K.v^{-1}=\{q \mid q.v \in K\}.$$
One can easily check that $[K.v] = [v][K], [K.v^{-1}]=[v^t][K]$ and $([K],1_n)=\frac{|K|}{n}$.
In order to simplify our notations we further omit square brackets.
Recall that a word $w$ is reset if and only if $q.w^{-1}$ for some state $q$ or equivalently $w^t q = [Q]$.
Let $P$ be any positive stochastic vector. Then $w$ is reset if and only if $([q.w^{-1}],P) = (w^t q, P)=1$.
It follows from $w^t q$ is a 1-0 vector and $P$ is positive.

Remark that one of the most fruitful method for finding quadratic upper bounds on the reset length is an \emph{extension} method.
In this method we choose some state $q$ and construct a finite sequence of words $w_1, w_2, \dots ,w_d$ such that
$$\frac{1}{n} = (q, 1_n) = ({w_1}^t q, 1_n) < ({w_2w_1}^t q, 1_n) < \dots < ({w_d \dots w_2w_1}^t q, 1_n)=1.$$
It is clear that such sequence can be constructed for any \san\ and its length $d$ is at most $n-1$ because
each inner product in the sequence exceeds previous for at least $\frac{1}{n}$.
Thus a quadratic upper bound will be proved as soon as one proves that the lengths of $w_i$ can be bounded by linear (in $n$) function.
For instance, if $|w_i| \leq n$ for $\mathrsfs{A}$ then it can be easily shown that the \v{C}ern\'{y} conjecture holds true for $\mathrsfs{A}$.
Using this fact the \v{C}ern\'{y} conjecture has been approved for \emph{circular}~\cite{Du98}, \emph{eulerian}~\cite{Ka03} and
\emph{one-cluster} automata with prime length cycle~\cite{SteinPrime10}. However, it is shown in~\cite{MyDLT2010} (see also the journal version~\cite{MyIJFCS10}) that there is a series of \sa\
where lengths of $w_i$ can not be bounded by $cn$ for any $1<c<2$. This means that for some proper subset $x \subset Q$ inequality $(v^t x, 1_n) \leq (x, 1_n)$
holds true for each word $v$ of length at most $cn$. Therefore the \v{C}ern\'{y} conjecture can not be always achieved on this way.
This suggests an idea to find a stochastic positive vector $P$ such that for each proper subset $x \subset Q$
there exists a word $v$ of length at most $n$ such that $(v^t x,P) > (x,P)$. It turns out that the vector $\alpha = \alpha(S(\mathrsfs{A},p))$
(the steady state distribution of Markov chain associated with $\mathrsfs{A}$ and probability vector $p$) satisfies this property.

\begin{theorem}\label{th_extension} Let $x \in R^n$ such that $(x,\alpha)=0$ and $v \in \Sigma^*$ be a word of minimum
length such that $(v^t x, \alpha) > 0$. Then
\begin{enumerate}
    \item\label{th_extension_0} $\sum_{u \in \Sigma^{r}}{p(u)(u^t x, \alpha)}=0$ for any $r \in \mathbb{N}$;
    \item\label{th_extension_1} if $|u| < |v|$ then $(u^t x, \alpha)=0$;
    \item\label{th_extension_2} $|v| \leq dim(\Sigma^{\leq n-1} \alpha)-1 \leq n-1$.
\end{enumerate}
\end{theorem}
\begin{proof}
    Items~\ref{th_extension_0},\ref{th_extension_1} immediately follow from $S^r \alpha = \alpha$ for any $r \in \mathbb{N}$.
If \linebreak $|v| \geq dim(\Sigma^{\leq n-1} \alpha)$ then from item~\ref{th_extension_1} $(u^t x, \alpha)=(x, u\alpha)=0$ for every $u, |u| < dim(\Sigma^{\leq n-1} \alpha)$.
For $i \in \{1,2, \dots ,n\}$ define a subspace $U_i = < u \alpha \mid |u| \leq i-1 >$.
Then a chain $$<\alpha> = U_1 \leq U_2 \leq \dots \leq U_{n} = \Sigma^{\leq n-1} \alpha.$$
becomes constant since some $j \leq dim(U_{n}) \leq |v|$, i.e.
$$U_1 < U_2 < \dots < U_j = U_{j+1} = \dots = U_{n}.$$
Thus $(x, u \alpha) = 0$ for every $u, |u| \leq dim(U_n)=dim(U_j)$ whence $(x, g) = 0$ for each $g \in U_{dim(U_j)}$.
Since $j$ has been chosen minimal $dim(U_j) \geq j$ then $U_{j} \supseteq U_{dim(U_j)}$.
So $(x, g) = 0$ for each $g \in U_{j} = U_{j+1} = \dots $ and $j \leq dim(U_{n}) \leq |v|$.
Since $v \alpha \in U_{|v|}=U_{j}$ then $(x, v \alpha) = 0$ and this contradicts with $(v^t x, \alpha) > 0$.
\end{proof}

It is worth to mention that similar view to synchronization process as a probability process were early
studied by Jungers~\cite{Ju11} but in contrast of linear programming techniques by Jungers we use techniques
from Perron-Frobenius theory. Moreover, the main result of~\cite{Ju11} is a similar proposition like in above theorem.
But in contrast with Jungers result we have a fixed vector $\alpha$ here and thus obtain quadratic upper bound
for a new class of automata in the next section.

\section{Quasi Eulerian Automata}

In view of theorem~\ref{th_extension} the lengths of extension words (for $\alpha = \alpha(S(\mathrsfs{A},p))$ instead $1_n$)
are bounded by $n-1$. Unfortunately we have here a conjugate problem that the lengths of such
sequences is hard to bound in general, because if $(K_1, \alpha) < (K_2, \alpha)$ for 1-0 vectors $K_1,K_2$ then
its difference $(K_2-K_1, \alpha)$ can be less than $\frac{1}{n}$. However, for some classes of \sa\ we can
directly use this theorem. At first prove an auxiliary statement.
\begin{corollary}
\label{cor_max_denom}
    Let $\alpha = \alpha(S(\mathrsfs{A},p)) \in Q^n$ for some probability vector $p$ on $\Sigma$ and
$L \in N$ denotes the least common multiple of denominators of $\alpha$ components.
Then $\mathfrak{C}(\mathrsfs{A}) \leq 1 + (n-1)(L-2)$.
\end{corollary}
\begin{proof}
    At first note that if $x_1,x_2$ are 1-0 vectors and $(x_2, \alpha) > (x_1, \alpha)$
then $(x_2,\alpha) \geq (x_1,\alpha) + \frac{1}{L}$. Since $\mathrsfs{A}$ is synchronizing there exists
a state $q$ and a letter $a$ such that $|a^{-1}q|>1$.
Set $w_1 = a$ then $(w_1^t q, \alpha) \geq (q, \alpha) + \frac{1}{L} \geq \frac{2}{L}$.
Suppose $(w_1^t q, \alpha) < 1$. Let $x_1 = w_1^t q - |w_1^t q| 1_n$ and $w_2$ be a word of
minimum length with $({w_2}^t x_1, \alpha) > 0$.
Such word exists because $$(u^t x_1, \alpha) = (Q - |w_1^t q| 1_n, \alpha) = 1 - (w_1^t q, \alpha) > 0$$ for any \sw\ $u$.
In view of theorem~\ref{th_extension} $|w_2| \leq n-1$ and $(w_2^t w_1^t q,\alpha) \geq ( |w_1^t q| 1_n, \alpha) + \frac{1}{L} \geq \frac{3}{L}$.
Continue in this way we construct a \sw\ $w_d w_{d-1} \dots w_1$ where $|w_1|=1$ and $|w_2| \leq n$. Since we start from $\frac{2}{L}$ and each
step adds to inner product at least $\frac{1}{L}$ then $d \leq \frac{(1- \frac{2}{L})}{\frac{1}{L}} \leq L-2$.
Thus $\mathfrak{C}(\mathrsfs{A}) \leq 1 + (n-1) d \leq 1 + (n-1)(L-2)$ and the corollary is proved.
\end{proof}

An automaton is Eulerian if its underlying graph admits
an Eulerian directed path, or equivalently, it is strongly connected and
the in-degree of every vertex is the same as the out-degree (and hence is
the alphabet size). It is clear that $\mathrsfs{A}$ is Eulerian if and
only if $S(\mathrsfs{A},1_n)$ is doubly stochastic. Due to~\cite{Stein2010}
$\mathrsfs{A}$ is pseudo-Eulerian if we can find a probability $p$
such that $S(\mathrsfs{A},p)$ is doubly stochastic.

\begin{corollary}
    If $\mathrsfs{A}$ is Eulerian or pseudo-Eulerian then $\mathfrak{C}(\mathrsfs{A}) \leq 1 + (n-1)(n-2)$.
\end{corollary}
\begin{proof}
    By condition we can choose a probability vector $p$ on $\Sigma$ to
provide $S(\mathrsfs{A},p)$ is row stochastic. Then $\alpha = \alpha(S(\mathrsfs{A},p)) = 1_n$ and in view of corollary~\ref{cor_max_denom}
we obtain the desired result.
\end{proof}

Remark that the same bounds for Eulerian and
have been proved early by Kari~\cite{Ka03} and later generalized for pseudo-Eulerian
automata by Steinberg~\cite{Stein2010} using another techniques. However, we now show that techniques
suggested in this paper is more powerful in some sense.

\begin{proposition}
\label{prop_diff_values}
    Let $\alpha = \alpha(S(\mathrsfs{A},p))$ for some probability vector $p$ on $\Sigma$ and
for some $c > 0$ there are $n-c$ equal numbers in a set of $\alpha$ components. Then
 $\mathfrak{C}(\mathrsfs{A}) \leq 2^c (n-c+1)(n-1)$.
\end{proposition}
\begin{proof}
    Without loss of generality let $\alpha = (r_1, r_2, \dots , r_c, r, r, \dots ,r)^t$ and $K \subset Q$.
Let $f_i$ determine that $K_i = 1$ for $i \in \{1,2, \dots ,c\}$ and $f_{c_+}$ be a number of 1's in $K$
with index more than $c$. Then $(K,\alpha) = \sum_{i=1}^{c}f_i r_i + f_{c_+}r$ whence this value is determined
by a vector $f(K)=(f_1,f_2, \dots, f_c, f_{c_+})$ where $f_{c_+} \in \{0 \dots n-c\}$ and $f_i \in \{0,1\}$.
Hence there are at most $2^c (n-c+1)$ possible different values of $(K,\alpha)$ and the length of any extension chain (for $\alpha$)
can not exceed $2^c (n-c+1)$. In view of theorem~\ref{th_extension} we can choose words of length at most $n-1$ and thus
we obtain a desired bound.
\end{proof}

As a corollary of this proposition we can prove a quadratic upper bound on the reset length for a new class of \sa.
We call automaton $\mathrsfs{A}$ \emph{quasi-Eulerian} with respect to $c \in N$ if there is an Eulerian or pseudo-Eulerian ``component''
$E_c$ with enter state $s$ which contains $n-c$ states, i.e. only state $s$ can have incoming arrows from $Q \setminus E_c$
and rows of $S(\mathrsfs{A},p)$ which corresponds to vertices from $E_c - s$ are row stochastic for some $p$.

\begin{theorem}
\label{th_qeuler}
    if $\mathfrak{C}(\mathrsfs{A})$ is quasi-Eulerian with respect to $c \in N$ then $\mathfrak{C}(\mathrsfs{A}) \leq 2^c (n-c+1)(n-1)$.
\end{theorem}
\begin{proof}
By condition for appropriate probability vector $p$ on $\Sigma$ we can provide that
rows of matrix $S=S(\mathrsfs{A},p)$ corresponding to states in $E_c - s$ are stochastic.
In view of theorem~\ref{th_extension} $\alpha$ is a single positive solution of equation $(S-E)x = 0$.
It is easy to show that all entries of $\alpha$ which corresponds to states from $E_c$ will have the same value
whence we can apply proposition~\ref{prop_diff_values} to $\alpha, c$ and obtain the desired result.
\end{proof}

As an example of quasi-Eulerian we can consider automata $\mathrsfs{C}_n$ from \v{C}ern\'{y} series.
One can easily check that $\mathrsfs{C}_n$ is quasi-Eulerian for $c=1$ and thus upper bound $\mathfrak{C}(\mathrsfs{C}_n) \leq 2n(n-1)$
follows from theorem~\ref{th_qeuler}. Finally, let us express our hope that ideas suggested in this paper could
be useful for the general case.

\end{document}